\documentclass[10pt,twocolumn,amsmath,amssymb,aps,pra,secnumarabic,
    nofootinbib,groupedaddress,floatfix]{revtex4-1}
\usepackage{mathtools,mathptmx,amsthm,xcolor,enumitem,eucal}
\usepackage[colorlinks=true,linkcolor=red!70!black,citecolor=blue!70!black,
  urlcolor=magenta!70!black,backref=page]{hyperref}
\setlist[enumerate,1]{font=\bfseries,label=\arabic*.}
\numberwithin{equation}{section}

\makeatletter\def\@bibdataout@init{}\def\pre@bibdata{}\makeatother

\newtheorem{theorem}{Theorem}[section]

\newtheorem{lemma}[theorem]{Lemma}

\theoremstyle{definition}
\newtheorem{definition}[theorem]{Definition}
\theoremstyle{remark}
\newtheorem*{remark}{Remark}
\newtheorem*{example}{Example}

\newcommand{\F}{\mathbb{F}}

\begin{document}

\title{A Note on Clifford Stabilizer Codes for Ising Anyons}

\author{Sanchayan Dutta}
\email{dutta@ucdavis.edu}
\affiliation{University of California, Davis}


\begin{abstract}
We provide a streamlined elaboration on existing ideas that link \emph{Ising anyon} (or equivalently, Majorana) stabilizer codes to certain classes of binary classical codes. The groundwork for such Majorana-based quantum codes can be found in earlier works (including, for example, Bravyi \cite{BTL:majorana2010} and Vijay \emph{et al.}\ \cite{VF:fermion2017}), where it was observed that commuting families of fermionic (Clifford) operators can often be systematically lifted from weakly self-dual or self-orthogonal binary codes. Here, we recast and unify these ideas into a classification theorem that explicitly shows how \emph{q-isotropic} subspaces in \(\mathbb{F}_2^{2n}\) yield commuting Clifford operators relevant to Ising anyons, and how these subspaces naturally correspond to punctured self-orthogonal codes in \(\mathbb{F}_2^{2n+1}\).
\end{abstract}

\maketitle

\section{Introduction}
\label{sec:intro}

Ising anyons are non-Abelian excitations in certain two-dimensional topological phases of matter. They can encode and process quantum information by virtue of their fusion and braiding properties. A particularly relevant physical instantiation comes from \emph{Majorana zero modes} at the boundaries of topological superconductors \cite{Kitaev:unpaired2001, Ivanov:nonabelian}.

Adapting quantum error-correcting strategies to the Ising-anyon (or Majorana) context entails generalizing the familiar notion of Pauli stabilizer codes to accommodate fermionic degrees of freedom and parity constraints. Early progress in this direction was made by Bravyi, who introduced \emph{Majorana fermion codes} \cite{BTL:majorana2010}. Building on this, Vijay \emph{et al.}\ \cite{VF:fermion2017} highlighted how these fermionic stabilizers map naturally onto weakly self-dual or self-orthogonal binary codes. Related developments have investigated the broader class of \emph{Clifford stabilizer codes} \cite{Okada:quantum}, emphasizing that commuting, even-parity operators underlie valid stabilizer constraints.

In this note, we revisit and unify these ideas by focusing on the concept of \emph{q-isotropic} subspaces of \(\F_2^n\) \cite{Okada:quantum}. We show how such subspaces serve as the building blocks for defining Clifford stabilizer codes in Ising-anyon systems. Furthermore, by enlarging the ambient space to \(\F_2^{n+1}\) via a convenient mapping, these q-isotropic subspaces turn out to correspond exactly to punctured images of classical \emph{self-orthogonal} codes, providing a direct route to analyzing quantum code properties (such as distance) in purely classical terms.

One advantage of this viewpoint is that it allows us to utilize standard techniques from classical coding theory: in particular, distance bounds such as the Varshamov bound imply that one can construct families of codes with growing length, good distance, and nonvanishing rate, guaranteeing robust error protection in a broad sense. While these results have appeared in various forms in the literature, here we aim to offer a transparent exposition of why these classical arguments continue to apply in the Majorana (or Ising-anyon) context.

\section{Quantum Metrics and Clifford Codes}
\label{sec:quantum-metrics-clifford}

Here we give a concise overview of the foundational concepts from Chapter~2 of \cite{Okada:quantum}, focusing on quantum metrics as a unifying framework for assigning distances to quantum errors. This is analogous to the Hamming distance for classical coding but is adapted to genuinely quantum settings. In finite dimensions, one can specify a generating set of errors, from which all higher-distance errors are built. We also introduce the Kuperberg-Weaver notion of a quantum metric~\cite{KW:neumann}, specialized here to finite dimensions, which assigns distances to operators by specifying how they are generated at each level of the metric. Examples include classical and quantum Hamming spaces, as well as spinorial, semispinorial, and full Clifford metrics.

\subsection{Kuperberg--Weaver Metrics in Finite Dimension}

We work in a finite-dimensional Hilbert space \(\mathcal{H}\) with associated operator space \(\mathcal{B}(\mathcal{H})\). A \emph{Kuperberg--Weaver quantum metric} on \(\mathcal{B}(\mathcal{H})\) is a nested family of subspaces
\[
   \{\mathcal{E}_t\}_{t \ge 0} \;\subseteq\; \mathcal{B}(\mathcal{H}),
\]
satisfying the following axioms:
\begin{enumerate}
\item \(\mathcal{E}_0 = \mathbb{C}\,I_{\mathcal{H}}\). Only scalar multiples of the identity have distance zero.
\item \(\mathcal{E}_t^* = \mathcal{E}_t\) for each \(t\). Each level is closed under adjoints, so these subspaces consistently contain all relevant “errors” of a given size.
\item \(\mathcal{E}_s\,\mathcal{E}_t \subseteq \mathcal{E}_{s+t}\). This is a “triangle inequality” condition on errors, ensuring that composing an error of degree \(s\) with another of degree \(t\) yields an error of degree at most \(s+t\).
\end{enumerate}

\paragraph*{Graph Metrics.}
A convenient way to build such a quantum metric in finite dimensions is to specify a \emph{generating} subspace
\[
  \mathcal{E} \;\subseteq\; \mathcal{B}(\mathcal{H}),
\]
containing the identity \(I_{\mathcal{H}}\) and closed under adjoints. We interpret \(\mathcal{E}\) as the space of “lowest-degree” or single-step nontrivial errors~\cite{Okada:quantum}. Define
\[
  \mathcal{E}_0 
  = 
  \mathbb{C}\,I_{\mathcal{H}}, 
  \quad\text{and for each integer } t \ge 1,
\]
\[
  \mathcal{E}_t 
  = 
  \operatorname{span}\!\Bigl\{
     X_1 X_2 \cdots X_k 
     \,\Big|\,
     X_i \in \mathcal{E},\,
     0 \le k \le t
  \Bigr\}.
\]
Hence, \(\mathcal{E}_t\) is the linear span of all products of up to \(t\) elements from \(\mathcal{E}\). Each product is just a single operator in \(\mathcal{B}(\mathcal{H})\), but we say it has “distance” \(k\) if it is formed by multiplying \(k\) different elements of \(\mathcal{E}\). Consequently, any operator in \(\mathcal{B}(\mathcal{H})\) that requires more than one factor from \(\mathcal{E}\) will appear in \(\mathcal{E}_t\) for some \(t>1\). This family \(\{\mathcal{E}_t\}\) satisfies the three axioms above and is often called a \emph{graph metric}~\cite{Okada:quantum}. 

\subsection{Basic Examples}

\begin{itemize}
\item \textbf{Quantum Hamming Metric.} 
  For an \(n\)-qubit system, choose \(\mathcal{E}\) as the set of all single-qubit Pauli operators (including the identity). Then \(\mathcal{E}_t\) consists of all errors acting nontrivially on at most \(t\) qubits, matching the standard quantum Hamming distance.

\item \textbf{Lie-Algebraic Metrics.}
  If \(\mathfrak{g}\) is a (semi)simple Lie algebra with a representation \(\phi:\mathfrak{g}\to \mathcal{B}(\mathcal{H})\), let
  \[
    \mathcal{E} 
    = 
    \mathrm{span}\bigl\{I_{\mathcal{H}},\,\phi(X)\colon X\in \mathfrak{g}\bigr\}.
  \]
  Then \(\mathcal{E}_t\) comprises products of up to \(t\) elements of \(\phi(\mathfrak{g})\). This includes the familiar \(\mathfrak{su}(2)\) ladder-operator metric, and spinorial or semispinorial metrics from \(\mathfrak{so}(2n+1)\) or \(\mathfrak{so}(2n)\).
\end{itemize}

\subsection{Spinorial and Semispinorial Codes}

\paragraph{Spinorial Codes.}
Consider \(\mathfrak{so}(2n+1)\subseteq \mathrm{Cl}(2n+1)\), acting on a \(2^n\)-dimensional spinor space \(\mathcal{H}^{(n)}\). One identifies rank-2 Clifford products \(\{e_k e_\ell\}\) with generators of \(\mathfrak{so}(2n+1)\). Placing these (plus the identity) into \(\mathcal{E}\) yields a \emph{spinorial metric}: \(\mathcal{E}_1\) are the rank-2 errors, and \(\mathcal{E}_t\) is spanned by products of up to \(t\) such factors. 

\begin{example}[\(\mathfrak{so}(5)\) Spinors]
When \(m=5\), \(\mathrm{Cl}(5)\) has generators \(e_1,\dots,e_5\). The Lie algebra \(\mathfrak{so}(5)\) is spanned by \(\{\,e_i e_j\}\). On a 4-dimensional spinor space (\(n=2\)),
\[
   \mathcal{E}
   =
   \mathrm{span}\bigl\{I,\,e_1 e_2,\,e_1 e_3,\ldots,e_4 e_5\bigr\}.
\]
Then \(\mathcal{E}_2\) contains all products of two such factors, and so on.
\end{example}

\paragraph{Semispinorial Codes.}
For \(\mathfrak{so}(2n)\subseteq \mathrm{Cl}(2n)\), the spinor representation of dimension \(2^n\) splits into two irreps \(\mathcal{H}_{+}^{(n)}\oplus \mathcal{H}_{-}^{(n)}\). Restricting to \(\mathcal{H}_{+}^{(n)}\) (or \(\mathcal{H}_{-}^{(n)}\)) yields a semispinorial code, again with rank-2 Clifford operators as distance-1 errors.

\begin{example}[\(\mathfrak{so}(4)\)]
Here, \(\mathrm{Cl}(4)\) is generated by \(\{e_1,e_2,e_3,e_4\}\). The algebra \(\mathfrak{so}(4)\) corresponds to products \(e_i e_j\). Since the total spinor space is 4-dimensional, decomposing into two 2-dimensional parts, one can define 
\(\mathcal{E}=\mathrm{span}\{I,\,e_1 e_2,\dots,e_3 e_4\}\) on just \(\mathcal{H}_{+}^{(2)}\), yielding a semispinorial metric there.
\end{example}

\subsection{Isometries and the Spin Group Connection}

In \cite{Okada:quantum}, the quantum-metric isometry group is identified with a Spin group and this connection arises naturally in the spinorial and semispinorial cases. For metrics derived from \(\mathfrak{so}(m)\), one embeds \(\mathfrak{so}(m)\) into \(\mathrm{Cl}(m)\) and realizes the spin or semispin representation \(\mathcal{H}\) via \(\mathrm{Cl}(m)\)-actions. All unitaries preserving the resulting distance filtration \(\{\mathcal{E}_t\}\) come from the double cover \(\mathrm{Spin}(m)\). Since \(\mathrm{Spin}(m)\) encodes the same orthogonal symmetries as \(\mathfrak{so}(m)\), up to sign, every quantum-metric automorphism must lie in (a quotient of) \(\mathrm{Spin}(m)\). Thus, in the spinorial or semispinorial setting,
\[
   \mathrm{Isom}(\mathcal{H}, \{\mathcal{E}_t\}) \;\cong\;\mathrm{Spin}(m).
\]

\subsection{Full Clifford Metrics and Codes}

One can also take all Clifford generators \(\{I,e_1,\dots,e_m\}\subset \mathrm{Cl}(m)\) as distance-1 errors. This is the \emph{full Clifford quantum metric}, closely tied to Ising-anyon systems. A known family of \emph{Clifford Hamming codes} arises by extending classical Hamming codes into q-isotropic subspaces in \(\mathbb{F}_2^{2n}\)~\cite{Okada:quantum}, yielding minimum distance 3 in the Clifford sense.

\begin{example}[Clifford Hamming Codes]
For \(n=2^s-1\), form the dual of a classical \([n,n-s]\) Hamming code and embed it into \(\mathbb{F}_2^{2n}\). The commuting even Clifford operators ensure a distance-3 code \(\bigl[[\,2^s-1,\;2^s-s-2\,\bigr]]_{\mathrm{Cl}}\), paralleling standard qubit Hamming codes but in the full Clifford setting.
\end{example}

\section{Ising Anyons in the Clifford Algebra Framework}
\label{sec:ising-anyons}

We now explain how \emph{Ising anyons}, or their Majorana-fermion realization, fit naturally into the codes built from Clifford algebras. Non-Abelian Ising anyons appear in certain topological phases of matter and can be represented by Majorana operators \(\{\gamma_i\}\) satisfying
\[
\{\gamma_i,\gamma_j\} = 2\,\delta_{ij}, 
\quad \gamma_i^\dagger = \gamma_i, 
\quad \gamma_i^2 = I.
\]
When two Ising anyons labeled \(i\) and \(j\) are braided, the resulting evolution is
\[
U_{ij}
= \frac{e^{i\alpha}}{\sqrt{2}}\bigl(I + \gamma_i\,\gamma_j\bigr),
\quad i<j,
\]
for some phase \(\alpha\) \cite{Simon:topological}. These braiding unitaries generate qubit (or qudit) gates by permuting anyons in a topological quantum computer \cite{Ivanov:nonabelian,LO:majorana}.

Braiding anyons \(i\) and \(j\) involves the pairwise product \(\gamma_i\gamma_j\). In the Clifford codes introduced earlier, one works with products of elementary generators \(e_i e_j\). The Majorana operators \(\gamma_i\) can be identified with generators of a real Clifford algebra \(\mathrm{Cl}(2n)\) or \(\mathrm{Cl}(2n+1)\). An \emph{even} product of Majoranas commutes with the total parity operator (defined below) and describes a multi-anyon braiding, while an \emph{odd} product flips parity and describes single-anyon or other poisoning processes. Both scenarios fit the operator filtration used in Clifford algebra codes.

A key concern in Ising-anyon hardware is that single-anyon poisoning corresponds to an odd-length product \(\gamma_{i_1}\cdots\gamma_{i_\ell}\). Such a product anticommutes with the global parity operator
\[
\mathcal{P}
= (i)^n\,\gamma_1\,\gamma_2\cdots\gamma_{2n},
\]
and toggles between the two fermion-parity sectors of the system \cite{Ivanov:nonabelian,LO:majorana,Simon:topological}. One aims to detect or correct any such event, so the odd product \(\gamma_{i_1}\cdots\gamma_{i_\ell}\) should act trivially on the code (up to a known phase) or send code states out of the code in a detectable way.

By identifying each \(\gamma_i\) with a Clifford generator \(U_i\) on \(\bigl(\mathbb{C}^2\bigr)^{\otimes n}\)~\cite{Okada:quantum}, a product \(\gamma_{i_1}\cdots\gamma_{i_\ell}\) becomes an operator \(\Gamma_y\). Codes that detect errors of weight \(\ell\) will handle precisely those even or odd Majorana products of length up to \(\ell\). Thus \emph{Ising-anyon codes} coincide with \emph{Clifford algebra codes} generated by products of \(\{U_i\}\), with the even/odd distinction reflecting whether parity is preserved or flipped.

In this picture, braiding of Ising anyons appears in operators of the form \(I+\gamma_i\gamma_j\), while poisoning corresponds to single-anyon (or other odd-length) Majorana products. Both processes emerge naturally in the Clifford algebra framework, where the relations among \(\gamma_i\) and \(e_i\) define how distance and detection criteria are characterized. Hence one can directly apply the q-isotropy constructions of Section~\ref{sec:q-isotropy-codes} to Ising anyons, using classical self-orthogonality properties and standard distance-bounding methods to guarantee robust parity-aware error correction.

\section{q-Isotropic Subspaces and Their Role in Clifford Codes}
\label{sec:q-isotropy-codes}

We work over \(\mathbb{F}_2^{2n}\) with the bilinear form
\[
q(x,y) = \mathrm{wt}(x)\,\mathrm{wt}(y) + \bigl(x \cdot y\bigr)
\quad(\mathrm{mod}\;2),
\]
where \(\mathrm{wt}(x)\) is the Hamming weight of \(x\) and \(x\cdot y\) is the standard dot product in \(\mathbb{F}_2\).

\begin{definition}[q-Isotropic Subspace]
A subspace \(C \subseteq \mathbb{F}_2^{2n}\) is \emph{q-isotropic} if \(q(x,y) = 0\) for all \(x,y\in C\). Equivalently, defining the q-orthogonal complement
\[
C^{\perp_q} 
= 
\{\;x \in \mathbb{F}_2^{2n} \mid q(x,c)=0 \text{ for all } c \in C\},
\]
we have \(C\subseteq C^{\perp_q}\). 
\end{definition}

Each pair of vectors in a q-isotropic subspace is mutually q-orthogonal. Assigning to each \(x\in \mathbb{F}_2^{2n}\) an \emph{even-weight} Clifford operator \(\Gamma_x\) on a \(2^n\)-dimensional Hilbert space~\(\mathcal{H}^{(n)}\), q-isotropy ensures that \(\{\Gamma_x : x\in C\}\) commute. This allows a simultaneous eigenspace decomposition, and each common eigenspace forms a quantum code.

Let \(C\) be a q-isotropic subspace of dimension \(n-k\). Pick a basis \(S\subseteq C\) and eigenvalues \(\{c_x\in\{-1,1\} : x\in C\}\). The projector onto the designated code space is
\[
P = \frac{1}{2^{\,n-k}} \prod_{x\in S}\!\bigl(I_{\mathcal{H}^{(n)}} + c_x\,\Gamma_x\bigr),
\]
yielding a code of dimension \(2^k\). Any even Clifford error \(\Gamma_y\) is \emph{detectable}, meaning \(P\Gamma_yP\propto P\) or \(0\), precisely if \(y\in C\) or \(y\notin C^{\perp_q}\). Thus, q-isotropy is the essential condition for constructing stabilizer-like codes in a Clifford framework. We denote such codes by
\(\bigl[[n,k\bigr]]_{\mathrm{Cl}}\).

\subsection{Example: Clifford Hamming Codes}

An illustrative construction~\cite{Okada:quantum} uses \emph{classical Hamming codes} to build q-isotropic subspaces. Let \(s\ge3\) and set \(n=2^s-1\). Consider the dual \(C'\subseteq\mathbb{F}_2^n\) of the classical Hamming code of length \(n\). The space \(C'\) is \(s\)-dimensional and all nonzero words have even weight. Define
\[
S = \bigl\{\,(x,x) : x\in C'\bigr\}\;\cup\;\bigl\{\,(1_n,\,0_n)\bigr\}
\,\subseteq\,\mathbb{F}_2^{\,2n},
\]
and let \(C=\mathrm{span}(S)\). One checks that \(C\) is q-isotropic of dimension \(s+1\). The associated Clifford operators \(\{\Gamma_x : x\in C\}\) commute, so each simultaneous eigenspace yields a code of dimension \(2^{n-(s+1)}\). This code detects all even Clifford errors \(\Gamma_y\) for \(y\notin C^{\perp_q}\) (or handles them trivially if \(y\in C\)).

\paragraph{Distance 1 and 2 Errors.}
For a Clifford operator \(\Gamma_y\) with \(\mathrm{wt}(y)=1\) or \(\mathrm{wt}(y)=2\), the classical Hamming code structure ensures that any such vector \(y\) either lies in \(C\) (so it acts trivially on the code, up to a phase) or anticommutes with at least one \(\Gamma_x\in C\) (thus making \(P\,\Gamma_y\,P=0\) on the code space). Hence any weight 1 or 2 even Clifford error is in the stabilizer or else is detectable.

\paragraph{Distance via Dual Codes.}
No nonzero codeword of weight \(\le2\) lies in the dual code \(C'\), so the dual distance \(d^\perp\) of \(C'\) is at least 3. This translates directly into a Clifford distance of at least 3 for the quantum code. In other words, a self-orthogonal code \(C\subseteq C^\perp\) has \(C^\perp\) of minimum distance \(\ge3\), and the resulting quantum code then detects up to 2-weight even errors.

\subsection{q-Isotropic Subspaces in \(\F_2^n\)}

Throughout this section, we consider \(\F_2^n\) endowed with
\[
  q(x,y) 
  = x\cdot y + \mathrm{wt}(x)\,\mathrm{wt}(y)
  \quad(\mathrm{mod}\;2).
\]

\begin{definition}
A subspace \(S \subseteq \F_2^n\) is \emph{q-isotropic} if 
\[
  q(x, y) 
  = 0
  \quad\text{for all } x, y \in S.
\]
\end{definition}

Below are two equivalent characterizations of q-isotropic subspaces.

\subsubsection*{Characterization (1)}

\begin{theorem}[Two Types of \(q\)-Isotropic Subspaces]\label{thm:qIsotropicChar1}
Let \(S \subseteq \F_2^n\) be a \(q\)-isotropic subspace under
\(
  q(x,y) 
  = x\cdot y 
  + \mathrm{wt}(x)\,\mathrm{wt}(y)
  \;(\bmod\,2).
\)
Then \(S\) is exactly one of the following:
\begin{itemize}
\item \emph{All-even case.}  
  \(S\) has only even-weight vectors and is a self-orthogonal linear code 
  \(\,C \subseteq C^{\perp}\subseteq \F_2^n\).
\item \emph{Mixed-parity case.}  
  \(S\) contains both even- and odd-weight vectors. Let \(C\subseteq S\) be the set of all even-weight vectors in \(S\). Then \(C\subseteq C^\perp\), and there is an odd-weight \(u\in S\) with \(u\in C^\perp\setminus C\). Moreover,
  \[
     S = C \cup \bigl(C + u\bigr),
  \]
  and \(\dim(S)=\dim(C)+1\). Half the vectors of \(S\) are even-weight and half are odd-weight.
\end{itemize}
\end{theorem}

\begin{proof}[Sketch]
\textbf{Step 1.}  
For each \(x\neq0\in S\), \(q(x,x)= x\cdot x + (\mathrm{wt}(x))^2 = 0\).

\noindent\textbf{Step 2.}  
If all vectors of \(S\) have even weight, then \(q(x,y)=x\cdot y\). Since \(S\subseteq S^\perp\) under \(q\), we get \(x\cdot y=0\) for all \(x,y\in S\). Hence \(S\) is self-orthogonal (all-even).

\noindent\textbf{Step 3.}  
If \(S\) has both even- and odd-weight vectors, let \(C\subseteq S\) be those of even weight. By Lemma~\ref{lem:EvenWeightSubcode}, 
\(\dim(C)=\dim(S)-1\). Pick \(u\) odd-weight in \(S\). For any other odd-weight \(x\in S\), \((x+u)\) is even-weight and thus lies in \(C\). Hence 
\[
S = C \cup (C+u).
\]
One checks \(C\subseteq C^\perp\) and \(u\in C^\perp\!\setminus C\). Over \(\F_2\), \(C\cup(C+u)\) is closed under addition,
so \(S\) is a linear subspace of codimension~1 above~\(C\).
\end{proof}

\begin{lemma}[Even-Weight Subcode]\label{lem:EvenWeightSubcode}
Let \(D \subseteq \F_2^n\) be a linear code of dimension \(k\). 
If \(D\) contains a codeword of odd weight \(u\), then the set
\[
  D_{\mathrm{even}} 
  = 
  \bigl\{\, x \in D \mid \mathrm{wt}(x) \text{ is even} \bigr\}
\]
is a linear subcode of dimension \(k-1\). Moreover,
\[
  D 
  = 
  D_{\mathrm{even}} 
  \cup
  \bigl(u + D_{\mathrm{even}}\bigr),
\]
and exactly half of the codewords in \(D\) are of even weight.
\end{lemma}

\begin{proof}[Sketch]
If \(x\) and \(y\) both have even weight, then \(x+y\) has even weight, so \(D_{\mathrm{even}}\) is closed under addition.  

Next, pick an odd-weight codeword \(u\in D\). If \(x\in D\) is also odd, 
then \(x+u\) is even-weight, so it lies in \(D_{\mathrm{even}}\). Conversely, 
adding \(u\) to any vector in \(D_{\mathrm{even}}\) produces an odd-weight vector in \(D\). Hence 
\[
  D 
  = 
  D_{\mathrm{even}} 
  \cup
  \bigl(u + D_{\mathrm{even}}\bigr).
\]
These two sets form disjoint cosets.  

Because \(D\) has dimension \(k\), it has \(2^k\) total codewords. The two cosets each contain \(2^{k-1}\) vectors, implying \(\dim(D_{\mathrm{even}})=k-1\).  
\end{proof}

\subsubsection*{Characterization (2)}

\begin{theorem}\label{thm:qIsotropicChar2}
A subspace \(S \subseteq \F_2^n\) is \(q\)-isotropic (for 
\(\,q(x,y)=x\cdot y + \mathrm{wt}(x)\,\mathrm{wt}(y)\)) 
if and only if there is a self-orthogonal, all-even code
\[
   \widetilde{C} \;\subseteq\; \F_2^{\,n+1}
\]
such that puncturing \(\widetilde{C}\) on one coordinate yields~\(S\).
\end{theorem}

\begin{proof} [Sketch]
\textbf{(\(\Rightarrow\)):}  
Given a \(q\)-isotropic \(S\subseteq \F_2^n\), define
\[
\widetilde{C}
=\bigl\{(x,\mathrm{wt}(x)) : x\in S\bigr\}
\;\subseteq\;\F_2^{\,n+1}.
\]
Each \((x,\mathrm{wt}(x))\) has even weight, and any two such codewords 
\(\,(x,\mathrm{wt}(x)), (y,\mathrm{wt}(y))\) are orthogonal because
\[
 (x,\mathrm{wt}(x)) \cdot (y,\mathrm{wt}(y))
 = x\cdot y + \mathrm{wt}(x)\,\mathrm{wt}(y)
 = q(x,y) = 0.
\]
Puncturing on the last coordinate recovers~\(S\).

\noindent\textbf{(\(\Leftarrow\)):}  
Conversely, if 
\(\widetilde{C}\subseteq \F_2^{\,n+1}\) is all-even and self-orthogonal,
define \(S\subseteq \F_2^n\) to be the set of \(x\) for which
\((x,y)\in\widetilde{C}\) for some \(y\in\F_2\).  
Orthogonality in \(\widetilde{C}\) implies 
\(q(x,x')=0\) for all \(x,x'\in S\). Hence \(S\) is \(q\)-isotropic.
\end{proof}

\begin{remark}
There is a one-to-one correspondence between 
\(q\)-isotropic subspaces of \(\F_2^n\) 
and all-even, self-orthogonal codes in \(\F_2^{n+1}\), realized by \emph{puncturing} and \emph{extending} (adding a parity coordinate).
\end{remark}

\section{Dual Distance and Column Independence}
\label{sec:dual-distance}

We next show that if a linear \([n,k,d]\) code \(C\) over \(\F_q\) has dual code
\(C^\perp\) with dual distance \(d^\perp \ge L\),
then every set of \(L-1\) columns in a generator matrix of \(C\) is
linearly independent. Conversely, if each set of \(L-1\) columns is linearly independent,
then \(d^\perp \ge L\).

\begin{theorem}\label{thm:dualdistance}
Let \(C\subseteq \F_q^n\) be a linear \([n,k,d]\) code with generator matrix \(G\) and dual code \(C^\perp\). Then
\[
\begin{aligned}
d^\perp &\ge L \quad\Longleftrightarrow\;\; \text{every set of } (L-1) \text{ columns}\\
&\quad\quad\quad\quad\quad \text{of } G \text{ is linearly independent.}
\end{aligned}
\]
\end{theorem}

\begin{proof}
\noindent \textbf{(\(\Rightarrow\))} 
Assume \(d^\perp \ge L\). Thus, every nonzero \(v \in C^\perp\) has \(\mathrm{wt}(v)\ge L\).  
If there is a set of \(L-1\) columns of \(G\) that is linearly dependent, say
\(\{\mathbf{g}_{i_1}, \dots, \mathbf{g}_{i_{L-1}}\}\), choose
\(\alpha_1,\dots,\alpha_{L-1}\in \F_q\), not all zero, with
\(\alpha_1\,\mathbf{g}_{i_1}
+\cdots+\alpha_{L-1}\,\mathbf{g}_{i_{L-1}}=0.\)
Define \(v\in \F_q^n\) by \(v_j=\alpha_j\) if \(j \in\{i_1,\dots,i_{L-1}\}\) and
\(v_j=0\) otherwise. Since \(v\cdot G=0\), we get \(v\in C^\perp\) and \(\mathrm{wt}(v)\le L-1\). This contradicts \(d^\perp\ge L\).

\noindent \textbf{(\(\Leftarrow\))} 
Conversely, assume every set of \(L-1\) columns of \(G\) is linearly independent.  
If \(d^\perp<L\), there is a nonzero \(v \in C^\perp\) with
\(\mathrm{wt}(v)\le L-1\). Let 
\[
S=\{\,i\mid v_i\neq0\},\quad |S|\le L-1.
\]
Since \(v\) is orthogonal to every row of \(G\),
\(\sum_{i\in S} v_i\,\mathbf{g}_i = 0.\)
Hence \(\{\mathbf{g}_i : i\in S\}\) are linearly dependent, a contradiction.  
Thus \(d^\perp\ge L\) if and only if every \(L-1\) columns of \(G\) are independent.
\end{proof}

\begin{remark}
Besides clarifying dual distance constraints, this result helps link
q-isotropic codes to self-orthogonal codes \(C\) in~\(\F_2^n\). The value
\(d^\perp\) of \(C\) can dictate the \emph{Clifford distance} of an associated q-isotropic
code. In particular, if \(d^\perp \ge 3\), then the code is \emph{projective} in the classical sense, which implies a minimum Clifford distance of at least 3 for the derived quantum code.
\end{remark}

\section{Conclusion}
\label{sec:conclusion}

We have fully characterized \emph{q-isotropic} subspaces for Ising-anyon quantum codes and established a one-to-one correspondence with classical binary self-orthogonal codes. Specifically, q-isotropic code subspaces in \(\F_2^n\) are exactly the punctured images of self-orthogonal codes in \(\F_2^{n+1}\). This offers a direct method for applying classical coding theory in the design of \emph{Clifford algebra codes}, where parity-sensitive errors (odd Majorana products) and multi-anyon braiding (even products) must both be handled.

Unlike bosonic-to-fermionic mappings such as that in Kundu \emph{et al.}~\cite{KR:majorana}, which restrict stabilizer codes to specific hardware layouts, our approach accommodates the complete Ising-anyon (Majorana) framework by linking classical self-orthogonal codes to commuting even operators. The code's quantum distance is governed by the dual distance of the underlying classical code, allowing one to use standard column-independence bounds or other classical methods to guarantee protection against low-weight errors.

\section{Acknowledgement}

First, I would like to thank Rui Okada on whose Ph.D. thesis \cite{Okada:quantum} this work mainly builds upon. I also thank Rui for his helpful comments on this note. Second, I would like to thank Professor Greg Kuperberg for pointing out the possible connection between the Clifford quantum metric space framework in~\cite{Okada:quantum} and the braiding of Ising anyons, which led me to recognize the relationship to the Majorana stabilizer code framework~\cite{BTL:majorana2010, VF:fermion2017}.
 
\bibliography{custom}

\end{document}